\documentclass[11pt]{article}
\usepackage[a4paper,margin=2.2cm]{geometry}

\newcommand{\red}[1]{{#1}\xspace}
\newcommand{\markRed}{}

\usepackage{latexsym}
\usepackage{xspace}
\usepackage{amssymb}
\usepackage{amsmath}
\usepackage{stmaryrd}
\usepackage{hyperref}
\usepackage{url}
\usepackage[all]{xy}
\usepackage{calc}
\usepackage{tikz}
\usepackage{flushend}
\usepackage{color}
\usepackage{enumerate}
\usepackage[linesnumbered,vlined,ruled]{algorithm2e}
\usepackage{authblk}

\def\tuple#1{\langle#1\rangle}

\def\eqref#1{(\ref{#1})}

\newcommand{\E}{\exists}

\def\false{\mathit{false}}

\newcommand{\mZ}{\mathcal{Z}}

\newcommand{\fALC}{$\mathit{f}\!\mathcal{ALC}$\xspace}

\newcommand{\myend}{\mbox{}\hfill{\footnotesize$\blacksquare$}}

\newcommand{\comment}[1]{}

\newcommand{\SV}{\Sigma_v}
\newcommand{\SE}{\Sigma_e}

\newcommand{\Next}{\mathit{Next}}
\newcommand{\Prev}{\mathit{Prev}}
\newcommand{\NextP}{\mathit{Next}'}
\newcommand{\PrevP}{\mathit{Prev}'}
\newcommand{\removeZ}{\textit{removed\_Z\_tp}}
\newcommand{\EdgeID}{\mathit{EdgeID}}
\newcommand{\EdgeIDp}{\mathit{EdgeID}'}
\newcommand{\Bool}{\mathit{bool}}
\newcommand{\Null}{\mathit{null}}

\newcommand{\rcNextP}{\mathit{rcNext}'}
\newcommand{\rcPrev}{\mathit{rcPrev}}
\newcommand{\rcNextElemP}{\mathit{rcNextElem}'}

\newcommand{\rcNext}{\mathit{rcNext}}
\newcommand{\rcPrevP}{\mathit{rcPrev}'}
\newcommand{\rcNextElem}{\mathit{rcNextElem}}

\SetKwInOut{Input}{Input}
\SetKwInOut{Output}{Output}
\SetKw{Break}{break}
\SetKw{Continue}{continue}

\newcommand{\ComputeSimulationEfficiently}{ComputeSimulationEfficiently}
\newcommand{\ComputeSimulation}{ComputeSimulation}

\newcommand{\ComputeDirectedSimulationEfficiently}{ComputeDirectedSimulationEfficiently}
\newcommand{\ComputeDirectedSimulation}{ComputeDirectedSimulation}

\newtheorem{theorem}{Theorem}[section]
\newtheorem{lemma}[theorem]{Lemma}
\newtheorem{proposition}[theorem]{Proposition}
\newtheorem{corollary}[theorem]{Corollary}

\newtheorem{Definition}[theorem]{Definition}
\newtheorem{Example}[theorem]{Example}
\newtheorem{Remark}[theorem]{Remark}

\newenvironment{example}{\begin{Example}\begin{em}}{\end{em}\end{Example}}
\newenvironment{remark}{\begin{Remark}\begin{em}}{\end{em}\end{Remark}}
\newenvironment{proof}{
	
	\smallskip
	
	\noindent
	{\em Proof.}}{
	
	\smallskip
	
}

\newcommand{\email}[1]{\mbox{Email: \url{#1}}}

\begin{document}
\sloppy
	
\title{Computing Crisp Simulations and Crisp Directed Simulations\\ for Fuzzy Graph-Based Structures}
		
\author{Linh Anh Nguyen}
		
\affil{
	Institute of Informatics, University of Warsaw, 
	Banacha 2, 02-097 Warsaw, Poland,  
	\email{nguyen@mimuw.edu.pl}
}


\date{(first version: August 10, 2020)}

\maketitle

\begin{abstract}
Like bisimulations, simulations and directed simulations are used for analyzing graph-based structures such as automata, labeled transition systems, linked data networks, Kripke models and interpretations in description logic. Simulations characterize the class of existential modal formulas, whereas directed simulations characterize the class of positive modal formulas. These notions are worth studying. For example, one may be interested in checking whether a given finite automaton simulates another or whether an object in a linked data network has all positive properties that another object has. To deal with vagueness and uncertainty, fuzzy graph-based structures are used instead of crisp ones. In this article, we design efficient algorithms with the complexity $O((m+n)n)$ for computing the largest crisp simulation and the largest crisp directed simulation between two finite fuzzy labeled graphs, where $n$ is the number of vertices and $m$ is the number of nonzero edges of the input fuzzy graphs. \red{We also adapt them to computing the largest crisp simulation and the largest crisp directed simulation between two finite fuzzy automata.} 
\end{abstract}


\section{Introduction}
\label{section:intro}

Bisimulations and simulations are used for characterizing similarity between states in automata, labeled transition systems and Kripke models. When applied to domains such as linked data networks or description logics, they can be used for characterizing similarity between objects or individuals. 
Phrases like ``two states are bisimilar'' or ``[something] simulates another'' are familiar in computer science. Searching Google Scholar with both the keywords ``bisimulation'' and ``simulation'' returns more than 13000 results. This shows their popularity and usefulness in computer science (see, e.g., \cite{vBenthem76,Park81,vBenthem83,vBenthem84,HennessyM85,BRV2001,Sangiorgi09}). 

Bisimulations have the following logical characterization in modal logic: modal formulas are invariant under bisimulations and, in image-finite Kripke models, two states are bisimilar (i.e., in the largest bisimulation relation) iff they cannot be distinguished by any modal formula (see, e.g., \cite{BRV2001}). There are different variants of modal logic and different corresponding notions of bisimulation. Besides, the second assertion (called the Hennessy-Milner property) can be made stronger by replacing ``image-finite'' with ``modal saturated''. The largest auto-bisimulation of a Kripke model is an equivalence relation.  

Simulations require fewer conditions than bisimulations and characterize only the class of existential modal formulas, which are formulas in negation normal form without universal modal operators. Namely, they do not require the ``backward'' condition(s) of bisimulations, existential modal formulas are preserved under simulations and, in image-finite Kripke models, a state $x'$ simulates a state $x$ iff $x'$ satisfies all existential modal formulas that $x$ satisfies (see, e.g., \cite{BRV2001}). The largest auto-simulation of a Kripke model is a pre-order.   

The notion of directed simulation was introduced by Kurtonina and de Rijke~\cite{KurtoninaR97} for modal logic. It was later formulated and studied by Divroodi and Nguyen~\cite{BSDL-P-LOGCOM} for description logics. Directed simulations are similar to bisimulations in that they use both the ``forward'' and ``backward'' conditions. The difference is that, if a state $x$ is directedly simulated by a state $x'$, then the label of $x$ is only required to be a subset of the label of $x'$, whereas two states in a bisimulation relation must have the same label. Directed simulations characterize the class of positive modal formulas~\cite{KurtoninaR97}. That is, positive modal formulas are preserved under directed simulations and, in image-finite Kripke models, a state $x'$ directedly simulates a state $x$ iff $x'$ satisfies all positive modal formulas that $x$ satisfies. The largest directed auto-simulation of a Kripke model is a~pre-order. 

The domains to which bisimulations, simulations and directed simulations are applied, such as automata, labeled transition systems, Kripke models, linked data networks and interpretations in description logic, are graph-based structures. To deal with vagueness and uncertainty, fuzzy graph-based structures are used instead, including fuzzy automata, fuzzy transition systems and weighted social networks.
In such structures, both labels of vertices (states or individuals) and labels of edges (transitions or connections) can be fuzzified. 
Extending the notions of bisimulation, simulation and directed simulation for fuzzy graph-based structures, we can consider crisp or fuzzy relations for them. Crisp bisimulations/simulations have been studied for fuzzy transition systems~\cite{CaoCK11,DBLP:journals/fss/WuD16,CaoSWC13,DBLP:journals/fss/WuCBD18}, weighted automata~\cite{StanimirovicSC2019}, fuzzy modal logics~\cite{Fan15} and fuzzy description logics~\cite{Nguyen-TFS2019,jBSfDL2}. Fuzzy bisimulations/simulations have been studied for fuzzy automata~\cite{CiricIDB12,CiricIJD12,MicicJS18}, weighted/fuzzy social networks~\cite{ai/FanL14,IgnjatovicCS15}, fuzzy modal logics~\cite{Fan15} and fuzzy description logics~\cite{jBSfDL2,minimization-by-fBS,TFS2020}. 
As shown in~\cite{Fan15,jBSfDL2}, the logical characterization of crisp bisimulations differs from the logical characterization of fuzzy bisimulations (in fuzzy modal/description logics under the G\"odel semantics) in that it uses a logical language extended with involutive negation or the Baaz projection operator. A similar claim can be stated for the difference between crisp simulations (respectively, crisp directed simulation) and fuzzy simulations (respectively, fuzzy directed simulation)~\cite{jBSfDL3}. 

The objective of this article is to design efficient algorithms for computing crisp simulations and crisp directed simulations between fuzzy graph-based structures. As related works, the closest ones are discussed below.
\begin{itemize}
\item In~\cite{PaigeT87} Paige and Tarjan gave an efficient algorithm with the complexity $O((m+n)\log{n})$ for computing the coarsest partition of a finite crisp graph, where $n$ is the number of vertices and $m$ is the number of edges of the graph. This problem is of the same nature as the problem of computing the largest auto-bisimulation of a finite crisp graph. Bloom and Paige~\cite{BloomP95} and Henzinger {\em et al.}~\cite{HenzingerHK95} gave algorithms with the complexity $O((m + n)n)$ for computing the largest auto-simulation of a crisp labeled transition system or a crisp graph, respectively. 

\item Adapting the idea of the above mentioned algorithm of Henzinger {\em et al.}~\cite{HenzingerHK95}, in~\cite{CompBSDLP} we gave an algorithm with the complexity $O((m + n)n)$ for computing the largest directed auto-simulation of a finite crisp graph. Furthermore, in that paper we also gave algorithms with the complexity \mbox{$O((m + n)^2 n^2)$} for computing the largest auto-simulation and the largest directed auto-simulation of a finite crisp graph for the setting with counting successors.

\item In~\cite{StanimirovicSC2019} Stanimirovi{\'c} {\em et at.}\ gave algorithms with the complexity $O(n^3)$ for computing the greatest right/left invariant Boolean equivalence matrix of a weighted automaton over an additively idempotent semiring. Such matrices are closely related to crisp bisimulations. They also gave algorithms with the complexity $O(n^5)$ for computing the greatest right/left invariant Boolean quasi-order matrix of a~weighted automaton over an additively idempotent semiring. Such matrices are closely related to crisp simulations. 

\item In~\cite{TFS2020} Nguyen and Tran gave algorithms with the complexity $O((m + n)n)$ for computing the greatest fuzzy bisimulation/simulation between two finite fuzzy interpretations in the fuzzy description logic \fALC under the G\"odel semantics, where $n$ is the number of individuals and $m$ is the number of nonzero instances of roles in the given fuzzy interpretations. They also adapted the algorithms to computing the greatest fuzzy bisimulation/simulation between fuzzy finite automata, as well as for dealing with other fuzzy description logics.
\end{itemize}

In this article, we design efficient algorithms with the complexity $O((m+n)n)$ for computing the largest crisp simulation and the largest crisp directed simulation between two finite fuzzy labeled graphs, where $n$ is the number of vertices and $m$ is the number of nonzero edges of the input fuzzy graphs. 
The motivations of this work are as follows:
\begin{itemize}
\item The research problem is worth studying. Given two finite fuzzy automata, we may be interested in checking whether one crisply simulates the other. This is related to computing the largest crisp simulation between the fuzzy automata. As another potential application, given two objects in a fuzzy linked data network, we may wonder whether one has all positive properties with a greater or equal degree than the other has. This is related to computing the largest crisp directed auto-simulation of the fuzzy network. 

\item As far as we know, there were no algorithms directly formulated for computing the largest crisp simulation and the largest crisp directed simulation between two finite fuzzy graph-based structures (such as fuzzy automata, fuzzy labeled transition systems, fuzzy Kripke models, and fuzzy interpretations in description logics). One can try to adapt the earlier mentioned algorithm with the complexity $O(n^5)$ by Stanimirovi{\'c} {\em et at.}~\cite{StanimirovicSC2019} to computing the largest crisp simulations and the largest crisp directed simulations, but the complexity order $O(n^5)$ is too high. 
\red{One can also crisp a given fuzzy labeled graph by treating a fuzzy $r$-labeled edge $\tuple{x,y}$ with a degree $d \in (0,1]$ as the set of all crisp $\tuple{r,d_i}$-labeled edges $\tuple{x,y}$, where $d_i \in (0,d]$ is any fuzzy value occurring in the specification of the input graph, and then apply one of the algorithms given in~\cite{BloomP95,HenzingerHK95,CompBSDLP} to the obtained crisp graph. The complexity of the resulting algorithm is of order $O(l(m+n)n)$, where $l$ is the number of fuzzy values occurring in the specification of the input fuzzy graph. In the worst case, $l$ can be $m+n$ and $O(l(m+n)n)$ is the same as $O((m+n)^2n)$, which is still too high.}  

\item \red{We choose to formulate our algorithms for fuzzy labeled graphs because the notion of fuzzy labeled graphs is universal. It covers fuzzy labeled transition systems, fuzzy Kripke models and fuzzy interpretations in description logic. Fuzzy automata are also a special case of fuzzy labeled graphs (see Section~\ref{sec: HFKDS}).}  
\end{itemize}

The rest of this article is structured as follows. \red{In Section~\ref{section: prel}, we define fuzzy labeled graphs, define crisp simulations and crisp directed simulations between such graphs, and present some properties.} In Section~\ref{section: comp-CS}, we present our algorithm of computing the largest crisp simulation between two finite fuzzy labeled graphs, prove it correctness and analyze its complexity. In Section~\ref{section: comp-CDS}, we extend and adapt that algorithm to computing the largest crisp directed simulation between two finite fuzzy labeled graphs. \red{In Section~\ref{sec: HFKDS}, we adapt our algorithms to computing the largest crisp simulation and the largest crisp directed simulation between two finite fuzzy automata.} We conclude the article in Section~\ref{sec: conc}. 

\section{Preliminaries}
\label{section: prel}

Fix a finite set $\SV$ of vertex labels and a finite set $\SE$ of edge labels. 
A {\em fuzzy labeled graph}, hereafter called a {\em fuzzy graph} for short, is a triple $G = \tuple{V, E, L}$, where $V$ is a set of vertices, \mbox{$E: V \times \SE \times V \to [0,1]$} is called the fuzzy set of labeled edges, and $L: V \to (\SV \to [0,1])$ is called the labeling function of vertices. 
It is {\em finite} if $V$ is finite. 
Given vertices $x,y \in V$, a vertex label $p \in \SV$ and an edge label $r \in \SE$, $L(x)(p)$ means the degree of that $p$ is a member of the label of~$x$, and $E(x,r,y)$ means the degree of that there is an edge labeled by~$r$ from $x$ to $y$. 

Given $f$ and $g$ of type \mbox{$\SV \to [0,1]$}, we write $f \leq g$ to denote that $f(p) \leq g(p)$ for all $p \in \SV$. 

Let $G = \tuple{V, E, L}$ and $G' = \tuple{V', E', L'}$ be fuzzy graphs. 
A binary relation $Z \subseteq V \times V'$ is called a {\em (crisp) simulation} between $G$ and $G'$ if the following two conditions hold (for all possible values of the free variables), where $\to$ and $\land$ denote the usual crisp logical connectives: 
\begin{eqnarray}
Z(x,x') & \to & L(x) \leq L'(x') \label{eq: CS 1} \\[1ex]
Z(x,x') \land E(x,r,y) > 0 & \to & \E y' \in V'\,(Z(y,y') \land E(x,r,y) \leq E(x',r,y')). \label{eq: CS 2}
\end{eqnarray}
{\markRed
To make it clear, as the connectives $\to$ and $\land$ are crisp, the conditions mean that:
\begin{itemize}
\item[\eqref{eq: CS 1}] if $Z(x,x')$ holds, then $L(x) \leq L'(x')$;
\item[\eqref{eq: CS 2}] if $Z(x,x')$ holds and $E(x,r,y) > 0$, then there exists $y' \in V'$ such that $Z(y,y')$ holds and $E(x,r,y) \leq E(x',r,y')$.
\end{itemize}
}

A binary relation $Z \subseteq V \times V'$ is called a {\em (crisp) directed simulation} between $G$ and $G'$ if Conditions~\eqref{eq: CS 1} and~\eqref{eq: CS 2} and the following condition hold (for all possible values of the free variables): 
\begin{eqnarray}
Z(x,x') \land E(x',r,y') > 0 & \to & \E y \in V\,(Z(y,y') \land E(x',r,y') \leq E(x,r,y)). \label{eq: CS 3}
\end{eqnarray}

\begin{remark}\ 
\begin{enumerate}
\item We allow the empty binary relation to be a simulation (respectively, directed simulation) between~$G$ and~$G'$. This is for increasing simplicity in talking about the largest simulation (respectively, directed simulation) between $G$ and~$G'$. 
\item Some works in the literature on simulations use the condition \mbox{$Z(x,x') \to L(x) = L'(x')$} instead of Condition~\eqref{eq: CS 1}. Our choice of using~\eqref{eq: CS 1} does not reduce the generality, because one can choose ${\SV}\!_2$, $L_2$ and $L'_2$ so that $L(x) = L'(x')$ over $\SV$ iff $L_2(x) \leq L'_2(x')$ over ${\SV}\!_2$ for all $\tuple{x,x'} \in V \times V'$. 
\item The definition of (crisp) directed simulation differs from the definition of (crisp) bisimulation in that it uses Condition~\eqref{eq: CS 1} instead of the condition \mbox{$Z(x,x') \to L(x) = L'(x')$}.\\ 
\myend
\end{enumerate}
\end{remark}


\begin{proposition}\label{prop: HGDFJ}
Let $G$, $G'$ and $G''$ be fuzzy graphs and let $G = \tuple{V, E, L}$. 
\begin{enumerate}
\item The relation $Z = \{\tuple{x,x} \mid x \in V\}$ is a simulation (respectively, directed simulation) between $G$ and itself.
\item If $Z_1$ is a simulation (respectively, directed simulation) between $G$ and $G'$, and $Z_2$ is a simulation (respectively, directed simulation) between $G'$ and $G''$, then $Z_1 \circ Z_2$ is a simulation (respectively, directed simulation) between $G$ and~$G''$.
\item If $\mZ$ is a set of simulations (respectively, directed simulations) between $G$ and $G'$, then $\bigcup\mZ$ is also a~simulation (respectively, directed simulation) between $G$ and $G'$.
\item Every directed simulation between $G$ and $G'$ is also a simulation between $G$ and $G'$.
\end{enumerate}   
\end{proposition}

The proof of this proposition is straightforward. 

The following corollary follows from the third assertion of Proposition~\ref{prop: HGDFJ}. 

\begin{corollary}
The largest simulation (respectively, directed simulation) between arbitrary fuzzy graphs exists.
\end{corollary}

A (crisp) {\em auto-simulation} (respectively, {\em directed auto-simulation}) of $G$ is a simulation (respectively, directed simulation) between $G$ and itself. The following corollary immediately follows from Proposition~\ref{prop: HGDFJ}. 

\begin{corollary}
The largest auto-simulation (respectively, directed auto-simulation) of a fuzzy graph is a pre-order.
\end{corollary}

\begin{example}\label{example: JDJWO 0a}\markRed
Let $\SV = \{p\}$, $\SE = \{r\}$ and let $G$ and $G'$ be the fuzzy graphs specified and illustrated below:
\begin{itemize}
\item $G = \tuple{V, E, L}$, where $V = \{a,b,c,d\}$ and 
	\begin{itemize}
		\item $L(a)(p) = 0.8$,\ \ $L(b)(p) = 0.8$,\ \ $L(c)(p) = 0.7$,\ \ $L(d)(p) = 0.9$, 
		\item $E(a,r,b) = 0.7$,\ \ $E(b,r,c) = 0.6$,\ \ $E(b,r,d) = 0.7$,\ \ $E(c,r,d) = 0.5$,\ \ $E(d,r,b) = 0.6$, 
		\item $E(x,r,y) = 0$ for the other triples $\tuple{x,r,y}$ with $x,y \in V$; 
	\end{itemize}
\item $G' = \tuple{V', E', L'}$, where $V' = \{e,f\}$ and
	\begin{itemize}
	\item $L'(e)(p) = 0.8$,\ \ $L'(f)(p) = 0.9$,
	\item $E'(e,r,e) = 0.6$,\ \ $E'(e,r,f) = 0.7$,\ \ $E'(f,r,e) = 0.6$, 
	\item $E'(x,r,y) = 0$ for the other triples $\tuple{x,r,y}$ with $x,y \in V'$. 
	\end{itemize}
\end{itemize}

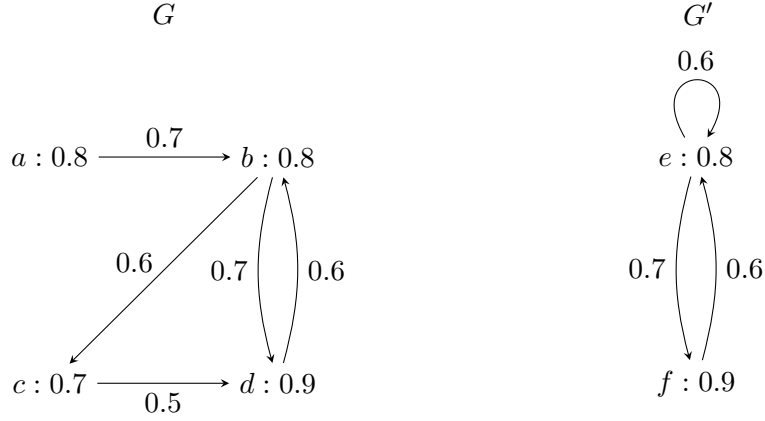
	\begin{center}
		\begin{tikzpicture}[->,>=stealth,auto]
		\node (G) {$G$};
		\node (Gp) [node distance=7cm, right of=G] {$\,G'$};
		\node (uG) [node distance=1.9cm, below of=G] {};
		\node (a) [node distance=1.5cm, left of=uG] {$a: 0.8$};
		\node (b) [node distance=1.5cm, right of=uG] {$b: 0.8$};
		\node (c) [node distance=3cm, below of=a] {$c: 0.7$};
		\node (d) [node distance=3cm, below of=b] {$d: 0.9$};
		\node (e) [node distance=1.9cm, below of=Gp] {$e: 0.8$};
		\node (f) [node distance=3cm, below of=e] {$f: 0.9$};
		\draw (a) to node{0.7} (b);
		\draw (b) to node[left,yshift=1mm]{0.6} (c);
		\draw (c) to node[below]{0.5} (d);
		\draw (b) edge [bend right=15] node[left]{0.7} (d);
		\draw (d) edge [bend right=15] node[right]{0.6} (b);
		\draw (e) edge [loop above,in=60,out=120,looseness=10] node{0.6} (e);
		\draw (e) edge [bend right=15] node[left]{0.7} (f);
		\draw (f) edge [bend right=15] node[right]{0.6} (e);
		\end{tikzpicture}
	\end{center}

It can be checked that $Z_0 = \{\tuple{b,e}$, $\tuple{c,e}$, $\tuple{d,f}\}$ is a simulation between $G$ and $G'$. 
Let $Z$ be the largest simulation between $G$ and $G'$. We show that $Z = Z_0$ by justifying that 
\[ \{\tuple{a,e}, \tuple{a,f}, \tuple{b,f}, \tuple{c,f}, \tuple{d,e}\} \cap Z = \emptyset. \] 
Observe the following. 
\begin{itemize}
\item $\tuple{d,e} \notin Z$ because $L(d) \not\leq L'(e)$.
\item $\tuple{c,f} \notin Z$ because \eqref{eq: CS 2} cannot hold for $\tuple{x,x',y} = \tuple{c,f,d}$ (since $\tuple{d,e} \notin Z$). 
\item $\tuple{b,f} \notin Z$ because \eqref{eq: CS 2} cannot hold for $\tuple{x,x',y} = \tuple{b,f,d}$. 
\item $\tuple{a,e} \notin Z$ because \eqref{eq: CS 2} cannot hold for $\tuple{x,x',y} = \tuple{a,e,b}$ (since $\tuple{b,f} \notin Z$). 
\item $\tuple{a,f} \notin Z$ because \eqref{eq: CS 2} cannot hold for $\tuple{x,x',y} = \tuple{a,f,b}$. 
\end{itemize}
Therefore, $\{\tuple{b,e}$, $\tuple{c,e}$, $\tuple{d,f}\}$ is the largest simulation between $G$ and $G'$. 
\myend
\end{example}

\begin{example}\label{example: JDJWO 0b}
Let $\SV$, $\SE$, $G$ and $G'$ be as in Example~\ref{example: JDJWO 0a}.
Now, let $Z$ be the largest directed simulation between $G$ and $G'$. We show that $Z = \emptyset$. 
Since every directed simulation between $G$ and $G'$ is also a simulation between $G$ and $G'$, by the claim of Example~\ref{example: JDJWO 0a}, $Z \subseteq \{\tuple{b,e}, \tuple{c,e}, \tuple{d,f}\}$. Hence, it suffices to show that 
\[ \{\tuple{b,e}, \tuple{c,e}, \tuple{d,f}\} \cap Z = \emptyset. \]
Observe the following. 
\begin{itemize}
\item $\tuple{c,e} \notin Z$ because \eqref{eq: CS 3} cannot hold for $\tuple{x,x',y'} = \tuple{c,e,e}$ (since $\tuple{d,e} \notin Z$). 
\item $\tuple{b,e} \notin Z$ because \eqref{eq: CS 3} cannot hold for $\tuple{x,x',y} = \tuple{b,e,e}$ (since $\{\tuple{c,e},\tuple{d,e}\} \cap Z = \emptyset$). 
\item $\tuple{d,f} \notin Z$ because \eqref{eq: CS 3} cannot hold for $\tuple{x,x',y'} = \tuple{d,f,e}$ (since $\tuple{b,e} \notin Z$). 
\end{itemize}
Therefore, $\emptyset$ is the largest directed simulation between $G$ and $G'$. 
\myend
\end{example}


\begin{example}\label{example: HJDJS}
Let $\SV = \{p\}$ and $\SE = \{r\}$. Let $G_2$ and $G_2'$ be the fuzzy graphs illustrated below and specified in a similar way as done for $G$ and $G'$ in Example~\ref{example: JDJWO 0a}. 
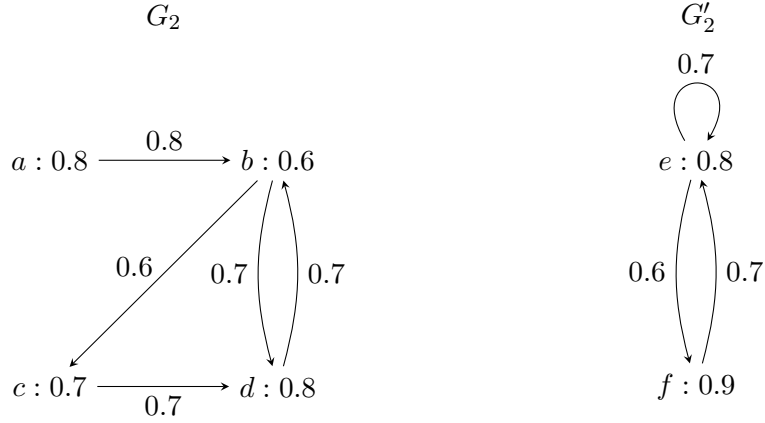
\begin{center}
	\begin{tikzpicture}[->,>=stealth,auto]
	\node (G) {$G_2$};
	\node (Gp) [node distance=7cm, right of=G] {$\,G_2'$};
	\node (uG) [node distance=1.9cm, below of=G] {};
	\node (a) [node distance=1.5cm, left of=uG] {$a: 0.8$};
	\node (b) [node distance=1.5cm, right of=uG] {$b: 0.6$};
	\node (c) [node distance=3cm, below of=a] {$c: 0.7$};
	\node (d) [node distance=3cm, below of=b] {$d: 0.8$};
	\node (e) [node distance=1.9cm, below of=Gp] {$e: 0.8$};
	\node (f) [node distance=3cm, below of=e] {$f: 0.9$};
	\draw (a) to node{0.8} (b);
	\draw (b) to node[left,yshift=1mm]{0.6} (c);
	\draw (c) to node[below]{0.7} (d);
	\draw (b) edge [bend right=15] node[left]{0.7} (d);
	\draw (d) edge [bend right=15] node[right]{0.7} (b);
	\draw (e) edge [loop above,in=60,out=120,looseness=10] node{0.7} (e);
	\draw (e) edge [bend right=15] node[left]{0.6} (f);
	\draw (f) edge [bend right=15] node[right]{0.7} (e);
	\end{tikzpicture}
\end{center}
It is straightforward to show that $Z = \{b,c,d\} \times \{e,f\}$ is the largest simulation (respectively, directed simulation) between $G_2$ and $G'_2$. 
\myend
\end{example}


\section{Computing Simulations}
\label{section: comp-CS}

In this section, we design an algorithm for computing the largest simulation between two given finite fuzzy graphs $G = \tuple{V, E, L}$ and $G' = \tuple{V', E', L'}$. 

For $r \in \SE$, $x,y \in V$ and $x',y' \in V'$, we denote 
\[
\begin{array}{rclrcl}
\Next_r(x) & \!\!=\!\! & \{y \in V \mid E(x,r,y) > 0\} & \qquad \Prev_r(y) & \!\!=\!\! & \{x \in V \mid E(x,r,y) > 0\} \\  
\NextP_r(x') & \!\!=\!\! & \{y' \in V' \mid E'(x',r,y') > 0\} & \PrevP_r(y') & \!\!=\!\! & \{x' \in V' \mid E'(x',r,y') > 0\}.
\end{array}
\]

\subsection{The Skeleton of the Algorithm}
\label{section: comp-CS1}

We first formulate our algorithm on an abstract level without implementation details. The aim is to facilitate understanding the skeleton of the algorithm and its correctness. Implementation details and complexity analysis will be presented in the next subsection.

It can be checked that Condition~\eqref{eq: CS 2} holds for all $r \in \SE$, $x,y \in V$ and $x' \in V'$ iff the following condition holds for all $r \in \SE$, $y \in V$ and $x' \in V'$, where the suprema are taken in the complete lattice [0,1]. 
\begin{eqnarray}\label{eq: CS 4}
\sup \{E(x,r,y) \mid x \in \Prev_r(y) \land \tuple{x,x'} \in Z\} \leq 
\sup \{E'(x',r,y') \mid y' \in \NextP_r(x') \land \tuple{y,y'} \in Z\}
\end{eqnarray}

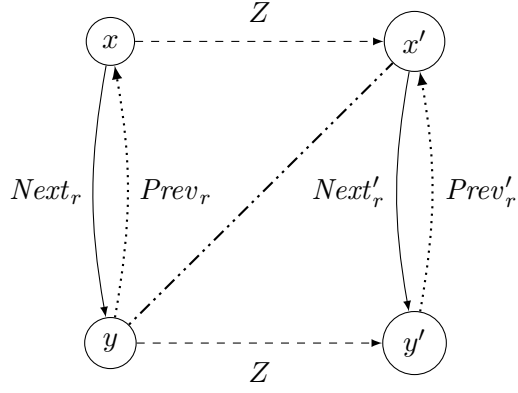
\begin{figure}
	\begin{center}
		\begin{tikzpicture}[auto,node distance=4cm]
		\tikzstyle{style}=[circle,draw]
		\node[style] (x) {$x$};
		\node[style] (y) [below of=x] {$y$};
		\node[style] (x2) [right of=x] {$x'$};
		\node[style] (y2) [right of=y] {$y'$};
		\draw[-latex] (x) edge [bend right=10, swap] node {$\Next_r\!$} (y);
		\draw[-latex,dotted, thick] (y) edge [bend right=10, swap] node {$\!\Prev_r$} (x);
		\draw[-latex] (x2) edge [bend right=10, swap] node {$\NextP_r\!$} (y2);
		\draw[-latex,dotted, thick] (y2) edge [bend right=10, swap] node {$\!\PrevP_r$} (x2);
		\draw[-latex,dashed] (x) edge node {$Z$} (x2);
		\draw[-latex,dashed] (y) edge node [swap] {$Z$} (y2);
		\draw[dash pattern={on 1pt off 2pt on 1pt off 2pt on 5pt off 2pt}, thick] (y) edge node {} (x2);
		\end{tikzpicture}
		\caption{An illustration for the algorithms given in Section~\ref{section: comp-CS}.\label{fig: dst1}}
	\end{center}
\end{figure}

Figure~\ref{fig: dst1} is helpful for illustrating Condition~\eqref{eq: CS 4}. It is worth emphasizing that this condition is involved with $\tuple{r,y,x'}$. Our algorithm constructs the largest relation $Z \subseteq V \times V'$ that satisfies Conditions~\eqref{eq: CS 1} and~\eqref{eq: CS 4} for all possible values of the free variables. 
It can be informally expressed as follows:
\begin{enumerate}
\item\label{step: HSHJH1} Initialize $Z$ by setting it to the largest subset of $V \times V'$ that satisfies Condition~\eqref{eq: CS 1}. 
\item Initialize $\removeZ$ by setting it to the empty set. This variable stands for the set of pairs removed from $Z$ and remained to be processed. So, $Z \cup \removeZ$ plays the role of a subset of a previous value of~$Z$. 

\item\label{step: HSHJH3} For each $\tuple{r,y,x'} \in \SE \times V \times V'$, move all the pairs $\tuple{x,x'} \in Z$ that falsify Condition~\eqref{eq: CS 4} to the set $\removeZ$. This can be restated as follows: for each $\tuple{r,y,x'} \in \SE \times V \times V'$ and $x \in \Prev_r(y)$, if $\tuple{x,x'} \in Z$ but the following condition does not hold, then move $\tuple{x,x'}$ from $Z$ to $\removeZ$. 
\begin{equation}\label{eq: CS 5}
E(x,r,y) \leq \sup \{E'(x',r,y') \mid y' \in \NextP_r(x') \land \tuple{y,y'} \in Z\}
\end{equation}

\item\label{step: HSHJH4} While $\removeZ \neq \emptyset$: 
	\begin{itemize} 
	\item extract $\tuple{y,y'}$ from $\removeZ$;
	\item for each $r \in \SE$, $x' \in \PrevP_r(y')$ and $x \in \Prev_r(y)$, if $\tuple{x,x'} \in Z$ but the following condition does not hold, then move $\tuple{x,x'}$ from $Z$ to $\removeZ$. 
	\begin{equation}\label{eq: CS 5b}
	E(x,r,y) \leq \sup \{E'(x',r,z') \mid z' \in \NextP_r(x') \land \tuple{y,z'} \in Z \cup \removeZ\}
	\end{equation}
	\end{itemize}
\end{enumerate}  

The steps~\ref{step: HSHJH1}--\ref{step: HSHJH3} together form the initialization, whereas the step~\ref{step: HSHJH4} is the main loop of the algorithm. The pseudocode is formulated as Algorithm~\ref{alg: ComputeSimulation} (on page~\pageref{alg: ComputeSimulation}).

\begin{figure*}[h]
\begin{procedure}[H]
\caption{ProcessPrev($r,y,x'$)\label{proc: ProcessPrev}\ \ // process $\Prev_r(y)$ with respect to $\tuple{r,y,x'}$}
$d := \sup \{E'(x',r,y') \mid y' \in \NextP_r(x') \textrm{ and } \tuple{y,y'} \in Z \cup \removeZ\}$\label{st: HANAE 1}\;
\ForEach{$x \in \Prev_r(y)$\label{st: HANAE 2}}{
	\If{$E(x,r,y) > d$ and $\tuple{x,x'} \in Z$}{ 
		move $\tuple{x,x'}$ from $Z$ to $\removeZ$\label{st: HANAE 4}\;
	}
}	
\end{procedure}

\begin{algorithm}[H]
\caption{\ComputeSimulation\label{alg: ComputeSimulation}}
\Input{finite fuzzy graphs $G = \tuple{V,E,L}$ and $G' = \tuple{V',E',L'}$.}
\Output{the largest simulation between $G$ and $G'$.}

\BlankLine
$Z := \{\tuple{x,x'} \in V \times V' \mid L(x) \leq L'(x')\}$\label{st: HGDSJ 1}\;
$\removeZ := \emptyset$\label{st: HGDSJ 2}\;
\ForEach{$\tuple{r,y,x'} \in \SE \times V \times V'$\label{st: HGDSJ 3}}{
	$\ProcessPrev(r,y,x')$\label{st: HGDSJ 4}\;
}
\BlankLine

\While{$\removeZ \neq \emptyset$\label{st: HGDSJ 5}}{
	extract $\tuple{y,y'}$ from $\removeZ$\;
	\ForEach{$r \in \SE$ and $x' \in \PrevP_r(y')$}{
		$\ProcessPrev(r,y,x')$\label{st: HGDSJ 8}\;
	}
}
\Return $Z$\;
\end{algorithm}
\end{figure*}

\begin{example}\label{example: JDJWO}
Let $\SV = \{p\}$, $\SE = \{r\}$ and let $G$ and $G'$ be the fuzzy graphs specified in Example~\ref{example: JDJWO 0a}.
Consider the execution of Algorithm~\ref{alg: ComputeSimulation} for $G$ and $G'$.  
	\begin{itemize}
		\item After executing the statement~\ref{st: HGDSJ 2}, we have that $Z = (\{a,b,c\} \times \{e,f\}) \cup \{\tuple{d,f}\}$ and $\removeZ = \emptyset$.
		\item Executing the ``foreach'' loop specified by the statements~\ref{st: HGDSJ 3} and~\ref{st: HGDSJ 4}, for the iteration involved with $\tuple{r,y,x'}$, 
		\begin{itemize}
			\item if $\tuple{r,y,x'} \in \{\tuple{r,a,e}$, $\tuple{r,a,f}$, $\tuple{r,b,e}\}$, then no changes are made, 
			\item if $\tuple{r,y,x'} = \tuple{r,b,f}$, then $\tuple{a,f}$ is moved from $Z$ to $\removeZ$, 
			\item if $\tuple{r,y,x'} \in \{\tuple{r,c,e}$, $\tuple{r,c,f}$, $\tuple{r,d,e}\}$, then no changes are made, 
			\item if $\tuple{r,y,x'} = \tuple{r,d,f}$, then $\tuple{b,f}$ and $\tuple{c,f}$ are moved from $Z$ to $\removeZ$.   
		\end{itemize}
		\item Thus, before executing the statement~\ref{st: HGDSJ 5}, we have that 
		\begin{eqnarray*}
			Z & = & (\{a,b,c\} \times \{e\}) \cup \{\tuple{d,f}\} \\
			\removeZ & = & \{\tuple{a,f},\tuple{b,f},\tuple{c,f}\}.
		\end{eqnarray*}
		
		\item Executing the ``while'' loop, for the iteration involved with $\tuple{y,y'}$ extracted from $\removeZ$, 
		\begin{itemize}
			\item if $\tuple{y,y'} = \tuple{a,f}$, then no additional changes are made, 
			\item if $\tuple{y,y'} = \tuple{b,f}$, then $\tuple{a,e}$ is moved from $Z$ to $\removeZ$, 
			\item if $\tuple{y,y'} \in \{\tuple{c,f}, \tuple{a,e}\}$, then no additional changes are made.
		\end{itemize}
		\item The algorithm returns $Z = \{\tuple{b,e},\tuple{c,e},\tuple{d,f}\}$. 
		This coincides with the claim of Example~\ref{example: JDJWO 0a} that this relation is the largest simulation between~$G$ and~$G'$.  
		\myend
	\end{itemize}
\end{example}

\begin{lemma}\label{lemma: JHDKA}
The following assertions are invariants of the ``while'' loop of Algorithm~\ref{alg: ComputeSimulation}: 
\begin{enumerate}
\item The largest simulation between $G$ and $G'$ is a subset of $Z$.
\item For all $\tuple{r,y,x'} \in \SE \times V \times V'$, if $x \in \Prev_r(y)$ and $Z[x,x']$ holds, then \eqref{eq: CS 5b} holds.
\end{enumerate}
\end{lemma}

\begin{proof}
Let $Z_0$ be the largest simulation between $G$ and $G'$. It must satisfy Condition~\eqref{eq: CS 4}, i.e., 
\begin{eqnarray}
\sup \{E(x,r,y) \mid x \in \Prev_r(y) \land \tuple{x,x'} \in Z_0\} \leq 
   \sup \{E'(x',r,y') \mid y' \in \NextP_r(x') \land \tuple{y,y'} \in Z_0\}. \label{eq: CS 6}
\end{eqnarray}
The first assertion is an invariant because it holds after the initialization (i.e., after executing the statement~\ref{st: HGDSJ 1}) and the removal of any pair from $Z$ at a later step is justifiable by using \eqref{eq: CS 6} and the induction assumption (which means that after such a removal it still holds that $Z_0 \subseteq Z$). 
It is straightforward to check that the second assertion is also an invariant of the ``while'' loop of Algorithm~\ref{alg: ComputeSimulation}. 
\myend
\end{proof}

\begin{theorem}\label{theorem: HDKAL}
Algorithm~\ref{alg: ComputeSimulation} always terminates and returns the largest simulation between the input fuzzy graphs.
\end{theorem}

\begin{proof}
After the initialization at the statement~\ref{st: HGDSJ 1}, no pairs are added to $Z$. The pairs added to $\removeZ$ are the ones extracted from $Z$. Hence, the total number of pairs added to $\removeZ$ is bounded. Each iteration of the ``while'' loop of Algorithm~\ref{alg: ComputeSimulation} extracts one pair from $\removeZ$. Hence, the loop and the algorithm itself always terminate. At the end, we have that $\removeZ = \emptyset$. Hence, by the second assertion of Lemma~\ref{lemma: JHDKA}, the resulting $Z$ satisfies Condition~\eqref{eq: CS 4} (for all possible values of the free variables). Due to the initialization, it also satisfies Condition~\eqref{eq: CS 1}. Thus, it is a simulation between $G$ and $G'$. By the first assertion of Lemma~\ref{lemma: JHDKA}, it follows that the returned $Z$ is the largest simulation between $G$ and $G'$.
\myend
\end{proof}

\subsection{Implementation Details and Complexity Analysis}
\label{section: comp-CS2}

By $|E|$ we denote $|\{\tuple{x,r,y} \in V \times \SE \times V : E(x,r,y) > 0\}|$, and similarly for $|E'|$. 
Let $n = |V| + |V'|$ and $m = |E| + |E'|$. Assume that $|\SV|$ and $|\SE|$ are constants. 
In this subsection, we give details on how to implement Algorithm~\ref{alg: ComputeSimulation} so that its complexity is of order $O((m+n)n)$. In particular, we provide an improved algorithm together with a complexity analysis. 

For simplicity, without loss of generality we assume that $V = 0\,..\,(|V|-1)$, $V' = 0..(|V'|-1)$ and $\SE = 0..(|\SE|-1)$. In practice, the input data for $G$ and $G'$ can be given in a friendly format using names for vertices and edge labels, but the conversion (using maps) for the input as well as for the output can be done in time $O((m+n)\log{n})$, which is within $O((m+n)n)$. 
 
From the input data we construct arrays $\NextP$, $\Prev$, $\PrevP$, $E$ and $E'$ such that, for $r \in \SE$, $x,y \in V$ and $x',y' \in V'$, $\NextP[r,x']$ is a vector representing $\NextP_r(x')$, $E[x,r,y] = E(x,r,y)$, and similarly for $\Prev[r,y]$, $\PrevP[r,y']$ and $E'[x',r,y']$. The construction can be done in time $O(n^2)$. 

We use an array $\EdgeIDp : V' \times \SE \times V' \to 0..(|E'|-1)$ for identifying nonzero edges of $G'$ so that, if $\tuple{x'_1,r_1,y'_1}$ and $\tuple{x'_2,r_2,y'_2}$ are different triples from $V' \times \SE \times V'$ such that $E'[x'_1,r_1,y'_1] > 0$ and $E'[x'_2,r_2,y'_2] > 0$, then $\EdgeIDp[x'_1,r_1,y'_1] \neq \EdgeIDp[x'_2,r_2,y'_2]$. 
Such an array can be constructed in time $O(n^2 + m\log{m})$. After having been constructed, $\EdgeIDp$ will never change. 

For the improved algorithm, we implement $Z$ as an array of type $V \times V' \to \Bool$ and $\removeZ$ as a queue. 
We also use the variables $\rcNextP$, $\rcNextElemP$ and $\rcPrev$ for representing arrays described below.

\begin{itemize}
\item $\rcNextP$ is an array such that, for $\tuple{r,y,x'} \in \SE \times V \times V'$, $\rcNextP[r,y,x']$ is a doubly linked list consisting of the vertices $y' \in \NextP_r(x')$ such that $Z[y,y'] \lor (\tuple{y,y'} \in \removeZ)$ holds. The list is sorted in ascending order with respect to $E'[x',r,y']$. The prefix ``rc'' stands for ``remaining for consideration''. The vertex contained in an element of $\rcNextP[r,y,x']$ is called the key of that element. 

\item $\rcNextElemP$ is an array with indices from $V \times (0..(|E'|-1))$ such that, if $Z[y,y'] \lor (\tuple{y,y'} \in \removeZ)$, $r \in \SE$, $x' \in \PrevP_r(y')$ and $e' = \EdgeIDp[x',r,y']$, then $\rcNextElemP[y,e']$ is (a reference to) the element of the doubly linked list $\rcNextP[r,y,x']$ whose key is~$y'$. 

\item $\rcPrev$ is an array such that, for $\tuple{r,y,x'} \in \SE \times V \times V'$, $\rcPrev[r,y,x']$ is a vector consisting of the vertices $x \in \Prev_r(y)$ such that the following condition (which is a reformulation of~\eqref{eq: CS 5b}) holds 
\begin{equation}
	E[x,r,y] \leq \sup \{E'[x',r,y'] \mid y' \in \NextP_r(x') \land (Z[y,y'] \lor (\tuple{y,y'} \in \removeZ))\} \label{eq: CS 7}
\end{equation}
and the vector is sorted in ascending order with respect to $E[x,r,y]$. 
\end{itemize}

\begin{figure*}
\begin{procedure}[H]
\caption{UpdateRcPrev($r,y,x'$)\label{proc: UpdateRcPrev}\ \ // for updating $\rcPrev$}

\lIf{$\rcNextP[r,y,x']$ is empty}{d := 0\label{st: HGNSA 1}}
\Else{
	let $y'$ be the key of the last element of $\rcNextP[r,y,x']$\;
	$d := E'[x',r,y']$\label{st: HGNSA 4}\;
}

\While{$\rcPrev[r,y,x']$ is not empty\label{st: HGNSA 5}}{
	let $x$ be the last element of $\rcPrev[r,y,x']$\label{st: HGNSA 6}\; 
	\lIf{$E[x,r,y] \leq d$}{\Break} 
	delete from $\rcPrev[r,y,x']$ the last element\;  
	\If{$Z[x,x']$}{
		$Z[x,x'] := \false$\;
		add $\tuple{x,x'}$ to $\removeZ$\label{st: HGNSA 11}\;
	}
}
\end{procedure}

\begin{procedure}[H]
\caption{Initialize()}
construct arrays $\NextP$, $\Prev$, $\PrevP$, $E$, $E'$ and $\EdgeIDp$ according to their specifications\label{st: HRLSN 1}\; 

\lForEach{$\tuple{x,x'} \in V \times V'$}{$Z[x,x'] := L(x) \leq L'(x')$\label{st: HRLSN 2}}
set $\removeZ$ to the empty queue\label{st: HRLSN 3}\;

\ForEach{$\tuple{r,x'} \in \SE \times V'$\label{st: HRLSN 4}}{
	sort $\NextP[r,x']$ in ascending order w.r.t.~$E'[x',r,y']$ for $y' \in \NextP[r,x']$\label{st: HRLSN 5}\;
}  
\ForEach{$\tuple{r,y} \in \SE \times V$\label{st: HRLSN 6}}{
	sort $\Prev[r,y]$ in ascending order w.r.t.~$E[x,r,y]$ for $x \in \Prev[r,y]$\label{st: HRLSN 7}\;
}  
\lForEach{$y \in V$ and $e' \in 0..(|E'|-1)$}{$\rcNextElemP[y,e'] := \Null$\label{st: HRLSN 8}} 

\ForEach{$\tuple{r,y,x'} \in \SE \times V \times V'$\label{st: HRLSN 9}}{
	construct a vector $\rcPrev[r,y,x']$ of all $x \in \Prev[r,y]$ such that $Z[x,x']$ holds, preserving the order\label{st: HRLSN 10}\;
	construct a doubly linked list $\rcNextP[r,y,x']$ whose elements' keys are all $y' \in \NextP[r,x']$ such that $Z[y,y']$ holds, preserving the order\label{st: HRLSN 11}\;
	\ForEach{element $u$ of $\rcNextP[r,y,x']$\label{st: HRLSN 12}}{
		let $y'$ be the key of $u$ and let $e' = \EdgeIDp[x',r,y']$\; 
		$\rcNextElemP[y,e'] := u$\label{st: HRLSN 14}\;
	}
} 

\ForEach{$\tuple{r,y,x'} \in \SE \times V \times V'$\label{st: HRLSN 15}}{
	$\UpdateRcPrev(r,y,x')$\label{st: HRLSN 16}\;
}
\end{procedure}

\begin{algorithm}[H]
\caption{\ComputeSimulationEfficiently\label{alg: ComputeSimulationEfficiently}}
\Input{finite fuzzy graphs $G = \tuple{V,E,L}$ and $G' = \tuple{V',E',L'}$.}
\Output{the largest simulation between $G$ and $G'$.}
\BlankLine
$\Initialize()$\;
\While{$\removeZ$ is not empty\label{st: JDNAA 2}}{
	extract $\tuple{y,y'}$ from $\removeZ$\label{st: JDNAA 3}\;
	\ForEach{$r \in \SE$ and $x' \in \PrevP[r,y']$\label{st: JDNAA 4}}{
		$e' := \EdgeIDp[x',r,y']$\label{st: JDNAA 5}\; 
		$u := \rcNextElemP[y,e']$\; 
		delete $u$ from $\rcNextP[r,y,x']$\;
		$\rcNextElemP[y,e'] := \Null$\label{st: JDNAA 8}\;
		$\UpdateRcPrev(r,y,x')$\label{st: JDNAA 9}\;
	}
}
\Return the relation corresponding to $Z$\label{st: JDNAA 10}\;
\end{algorithm}
\end{figure*}

Algorithm~\ref{alg: ComputeSimulationEfficiently} (on page~\pageref{alg: ComputeSimulationEfficiently}) is a reformulation of Algorithm~\ref{alg: ComputeSimulation} using the above described data structures. It is given together with its subroutines $\UpdateRcPrev$ (for updating $\rcPrev$) and $\Initialize$ on page~\pageref{alg: ComputeSimulationEfficiently}. The procedure $\UpdateRcPrev$ corresponds to the procedure $\ProcessPrev$ used in Algorithm~\ref{alg: ComputeSimulation}. 
Roughly speaking, the reformulation relies on the following two points.  
\begin{itemize}
\item Computing $d$ specified by the statement~\ref{st: HANAE 1} of the procedure $\ProcessPrev(r,y,x')$ used in Algorithm~\ref{alg: ComputeSimulation} is implemented in Algorithm~\ref{alg: ComputeSimulationEfficiently} by the statements~\ref{st: HGNSA 1}-\ref{st: HGNSA 4} of the procedure $\UpdateRcPrev(r,y,x')$ so that its cost is bounded by a constant. It is done by using the sorted list $\rcNextP[r,y,x']$. This list is doubly linked so that deleting any element from the list can be done in constant time. Such a deletion is triggered when a~pair $\tuple{y,y'}$ is extracted from $\removeZ$ by the statement~\ref{st: JDNAA 3} of Algorithm~\ref{alg: ComputeSimulationEfficiently}. In such a situation, to guarantee that $\rcNextP$ satisfies its specification, we need to delete the element~$u$ whose key is~$y'$ from the list $\rcNextP[r,y,x']$. A reference to $u$ is kept in $\rcNextElemP[y,e']$, where $e' = \EdgeIDp[x',r,y']$. The deletion is done by the statements~\ref{st: JDNAA 5}-\ref{st: JDNAA 8} of Algorithm~\ref{alg: ComputeSimulationEfficiently}. This shows the key role of the arrays $\rcNextP$, $\rcNextElemP$ and $\EdgeIDp$.  

\item When dealing with a triple $\tuple{r,y,x'} \in \SE \times V \times V'$, moving pairs $\tuple{x,x'}$ from $Z$ to $\removeZ$ as done in the loop specified by the statements \ref{st: HANAE 2}-\ref{st: HANAE 4} of the procedure $\ProcessPrev(r,y,x')$ (used in Algorithm~\ref{alg: ComputeSimulation}) is implemented in Algorithm~\ref{alg: ComputeSimulationEfficiently} by the statements~\ref{st: HGNSA 5}-\ref{st: HGNSA 11} of the procedure $\UpdateRcPrev(r,y,x')$. Once again, the task is realized by using a sorted collection ($\rcPrev[r,y,x']$). The cost related to each $x \in \Prev_r(y)$ of the task is bounded by a constant. This is the point of the use of the array $\rcPrev$.  
\end{itemize} 

We have implemented Algorithm~\ref{alg: ComputeSimulationEfficiently} in C++ and shared the code~\cite{compCSfFS-impl}. The reader can experiment with the algorithm by adding statements to the code for displaying the data he or she wants to watch.


\begin{lemma}\label{lemma: JHDKA 2}
The following assertions are invariants of the ``while'' loop of Algorithm~\ref{alg: ComputeSimulationEfficiently}: 
\begin{enumerate}
\item The data structures $\rcNextP$, $\rcNextElemP$ and $\rcPrev$ satisfy their specifications. 
\item The largest simulation between $G$ and $G'$ is a subset of $\{\tuple{x,x'} \in V \times V' \mid Z[x,x']\}$.
\item For all $\tuple{r,y,x'} \in \SE \times V \times V'$, if $x \in \Prev_r(y)$ and $Z[x,x']$ holds, then $x \in \rcPrev[r,y,x']$. 
\end{enumerate}
\end{lemma}

\begin{proof}
It is easy to check that the first assertion is an invariant of the ``while'' loop of Algorithm~\ref{alg: ComputeSimulationEfficiently}. 

Consider the second assertion and let $Z_0$ be the largest simulation between $G$ and $G'$. Recall that $Z_0$ satisfies~\eqref{eq: CS 6}. We prove that $Z_0 \subseteq \{\tuple{x,x'} \in V \times V' \mid Z[x,x']\}$ by induction on the step number during the execution of Algorithm~\ref{alg: ComputeSimulationEfficiently}. The base case occurs after the initialization of $Z$ by the statement~\ref{st: HRLSN 2} of the procedure $\Initialize$. The induction hypothesis clearly holds for the base case. For the induction step, assume that the induction hypothesis holds before calling the procedure $\UpdateRcPrev(r,y,x')$. We only need to prove that it still holds after executing that procedure. Let $d$ be the value computed by the statements \ref{st: HGNSA 1}-\ref{st: HGNSA 4} of $\UpdateRcPrev(r,y,x')$. 
By the first invariant, we have that 
\[
d = \sup \{E'(x',r,y') \mid y' \in \NextP_r(x') \land (Z[y,y'] \lor (\tuple{y,y'} \in \removeZ))\}.	 
\]
By the induction assumption, $Z_0 \subseteq \{\tuple{x,x'} \in V \times V' \mid Z[x,x']\}$. Hence, 
\[
d \geq \sup \{E'(x',r,y') \mid y' \in \NextP_r(x') \land \tuple{y,y'} \in Z_0\}.	 
\]
By~\eqref{eq: CS 6}, it follows that 
\[
d \geq \sup \{E(x,r,y) \mid x \in \Prev_r(y) \land \tuple{x,x'} \in Z_0\}.
\]
Therefore, after executing the statements \ref{st: HGNSA 5}-\ref{st: HGNSA 11} of $\UpdateRcPrev(r,y,x')$, it still holds that $Z_0 \subseteq \{\tuple{x,x'} \in V \times V' \mid Z[x,x']\}$. This completes the proof of that the second assertion is an invariant of the ``while'' loop of Algorithm~\ref{alg: ComputeSimulationEfficiently}. 

The third assertion is also an invariant of the loop because it holds after executing the statements \ref{st: HRLSN 1}-\ref{st: HRLSN 14} of the procedure $\Initialize$ and, whenever $x$ is deleted from $\rcPrev[r,y,x']$, $Z[x,x']$ is set to $\false$. 
\myend
\end{proof}

We give below the main theorem of this section. 

\begin{theorem}\label{theorem: JHLWA}
Algorithm~\ref{alg: ComputeSimulationEfficiently} runs in time $O((m+n)n)$ and returns the largest simulation between given finite fuzzy graphs $G = \tuple{V,E,L}$ and $G' = \tuple{V',E',L'}$, where $n = |V| + |V'|$ and $m = |E| + |E'|$. 
\end{theorem}

\begin{proof}
First, consider the complexity of the procedure $\Initialize$.
\begin{itemize}
\item As mentioned earlier, constructing the arrays $\NextP$, $\Prev$, $\PrevP$, $E$ and $E'$ can be done in time $O(n^2)$, whereas constructing the array $\EdgeIDp$ can be done in time $O(n^2 + m\log{m})$. Hence, the statement~\ref{st: HRLSN 1} (of the procedure) runs in time $O(n^2 + m\log{m})$. 
\item The statement~\ref{st: HRLSN 2} runs in time $O(n^2)$. The statement~\ref{st: HRLSN 3} runs in constant time.
\item The loops specified by the statements \ref{st: HRLSN 4}-\ref{st: HRLSN 7} run in time $O(n + m\log{m})$.  
\item The loop specified by the statement~\ref{st: HRLSN 8} runs in time $O(nm)$.  
\item The loop specified by the statements \ref{st: HRLSN 9}-\ref{st: HRLSN 14} runs in time $O(n(n+m))$.  
\item The loop specified by the statements~\ref{st: HRLSN 15} and~\ref{st: HRLSN 16} runs in time $O(n(n+m))$.  
\end{itemize}
Summing up, the procedure $\Initialize$ runs in time $O((m+n)n)$. 

Now consider the complexity of the ``while'' loop of Algorithm~\ref{alg: ComputeSimulationEfficiently}.
\begin{itemize}
\item The statements~\ref{st: JDNAA 3} and~\ref{st: JDNAA 5}-\ref{st: JDNAA 8} of Algorithm~\ref{alg: ComputeSimulationEfficiently} as well as the statements~\ref{st: HGNSA 1}-\ref{st: HGNSA 4} and~\ref{st: HGNSA 6}-\ref{st: HGNSA 11} of the procedure $\UpdateRcPrev(r,y,x')$ run in constant time. 
\item Each iteration of the ``foreach'' loop of Algorithm~\ref{alg: ComputeSimulationEfficiently} is involved with a pair $\tuple{y,y'}$ extracted from $\removeZ$ and a vertex $x' \in \PrevP_r(y')$. 

\item Each iteration of the ``while'' loop of the procedure $\UpdateRcPrev(r,y,x')$ is involved with the triple $\tuple{r,y,x'} \in \SE \times V \times V'$ and a vertex $x \in \Prev_r(y)$. 
If $x$ is deleted from $\rcPrev[r,y,x']$, then we ascribe the cost of the involved iteration to the edge $\tuple{x,r,y}$ and~$x'$. Otherwise, the involved iteration is the last one of the loop and we ascribe its cost to the pair $\tuple{y,y'}$ and $x' \in \PrevP_r(y')$ which together identify the iteration of the ``foreach'' loop of Algorithm~\ref{alg: ComputeSimulationEfficiently} that calls $\UpdateRcPrev(r,y,x')$.  
\end{itemize}
Thus, the ``while'' loop of Algorithm~\ref{alg: ComputeSimulationEfficiently} runs in time $O((m+n)n)$. 

The ``return'' statement of Algorithm~\ref{alg: ComputeSimulationEfficiently} runs in time $O(n^2)$. 

Summing up, Algorithm~\ref{alg: ComputeSimulationEfficiently} runs in time $O((m+n)n)$.   

At the end of the execution of Algorithm~\ref{alg: ComputeSimulationEfficiently}, $\removeZ$ is empty. 
By the first and third assertions of Lemma~\ref{lemma: JHDKA 2} (see, among others, \eqref{eq: CS 7}), it follows that $Z$ satisfies Condition~\eqref{eq: CS 4}. By the initialization, $Z$ also satisfies Condition~\eqref{eq: CS 1}. Hence, $Z$ is a simulation between $G$ and $G'$. Together with the second assertion of Lemma~\ref{lemma: JHDKA 2}, this implies that the relation returned by Algorithm~\ref{alg: ComputeSimulationEfficiently} is the largest simulation between~$G$ and~$G'$.
\myend
\end{proof}


\section{Computing Directed Simulations}
\label{section: comp-CDS}

In this section, we extend and adapt the algorithms given in the previous section to obtain algorithms for computing the largest directed simulation between two given finite fuzzy graphs $G = \tuple{V, E, L}$ and $G' = \tuple{V', E', L'}$. 
We use all notions and data structures introduced in the previous section. 

It can be checked that Condition~\eqref{eq: CS 3} holds for all $r \in \SE$, $x \in V$ and $x',y' \in V'$ iff the following condition holds for all $\tuple{r,x,y'} \in \SE \times V \times V'$. 
\begin{equation}\label{eq: GFDJS}
\sup \{E(x',r,y') \mid x' \in \PrevP_r(y') \land \tuple{x,x'} \in Z\} \leq 
\sup \{E(x,r,y) \mid y \in \Next_r(x) \land \tuple{y,y'} \in Z\}.
\end{equation}
In some sense, Conditions~\eqref{eq: CS 3} and~\eqref{eq: GFDJS} are dual to Conditions~\eqref{eq: CS 2} and~\eqref{eq: CS 4}, respectively. Algorithm~\ref{alg: ComputeDirectedSimulation} (\ComputeDirectedSimulation) given on page~\pageref{alg: ComputeDirectedSimulation} is our adapted extension of Algorithm~\ref{alg: ComputeSimulation} for computing the largest directed simulation between~$G$ and~$G'$. It takes into account the mentioned duality. In particular, 
\begin{itemize}
\item the procedure $\ProcessPrev'(r,x,y')$ is dual to the procedure $\ProcessPrev(r,y,x')$, 
\item the statements \ref{st: KWOAJ 1}-\ref{st: KWOAJ 4} and \ref{st: KWOAJ 7}-\ref{st: KWOAJ 10} of Algorithm~\ref{alg: ComputeDirectedSimulation} are the same as the statements \ref{st: HGDSJ 1}-\ref{st: HGDSJ 4} and \ref{st: HGDSJ 5}-\ref{st: HGDSJ 8} of Algorithm~\ref{alg: ComputeSimulation}, respectively, 
\item the loops specified by the statements~\ref{st: KWOAJ 5}-\ref{st: KWOAJ 6} and \ref{st: KWOAJ 11}-\ref{st: KWOAJ 12} of Algorithm~\ref{alg: ComputeDirectedSimulation} are dual to the loops specified by its statements~\ref{st: KWOAJ 3}-\ref{st: KWOAJ 4} and \ref{st: KWOAJ 9}-\ref{st: KWOAJ 10}, respectively. 
\end{itemize}

\begin{figure*}[h]
\begin{procedure}[H]
\caption{ProcessPrev$'$($r,x,y'$)\label{proc: ProcessPrev'}\ \ // process $\PrevP_r(y')$ with respect to $\tuple{r,x,y'}$}
$d := \sup \{E(x,r,y) \mid y \in \Next_r(x) \textrm{ and } \tuple{y,y'} \in Z \cup \removeZ\}$\label{st: HANAE 1b}\;
\ForEach{$x' \in \PrevP_r(y')$\label{st: HANAE 2b}}{
	\If{$E'(x',r,y') > d$ and $\tuple{x,x'} \in Z$}{ 
		move $\tuple{x,x'}$ from $Z$ to $\removeZ$\label{st: HANAE 4b}\;
	}
}	
\end{procedure}
	
\begin{algorithm}[H]
	\caption{\ComputeDirectedSimulation\label{alg: ComputeDirectedSimulation}}
	\Input{finite fuzzy graphs $G = \tuple{V,E,L}$ and $G' = \tuple{V',E',L'}$.}
	\Output{the largest directed simulation between $G$ and $G'$.}
	
	\BlankLine
	$Z := \{\tuple{x,x'} \in V \times V' \mid L(x) \leq L'(x')\}$\label{st: KWOAJ 1}\;
	$\removeZ := \emptyset$\;
	
	\ForEach{$\tuple{r,y,x'} \in \SE \times V \times V'$\label{st: KWOAJ 3}}{
		$\ProcessPrev(r,y,x')$\label{st: KWOAJ 4}\tcp*{defined on page~\pageref{proc: ProcessPrev}}
	}

	\ForEach{$\tuple{r,x,y'} \in \SE \times V \times V'$\label{st: KWOAJ 5}}{
		$\ProcessPrev'(r,x,y')$\label{st: KWOAJ 6}\;
	}
	
	\BlankLine
	\While{$\removeZ \neq \emptyset$\label{st: KWOAJ 7}}{
		extract $\tuple{y,y'}$ from $\removeZ$\;
		
		\ForEach{$r \in \SE$ and $x' \in \PrevP_r(y')$\label{st: KWOAJ 9}}{
			$\ProcessPrev(r,y,x')$\label{st: KWOAJ 10}\;
		}

		\ForEach{$r \in \SE$ and $x \in \Prev_r(y)$\label{st: KWOAJ 11}}{
			$\ProcessPrev'(r,x,y')$\label{st: KWOAJ 12}\;
		}
	}
	\Return $Z$\;
\end{algorithm}
\end{figure*}


\begin{example}\label{example: JDJWO 2}
Let $\SV = \{p\}$, $\SE = \{r\}$ and let $G$ and $G'$ be the fuzzy graphs specified in Example~\ref{example: JDJWO 0a}. 
Consider the execution of Algorithm~\ref{alg: ComputeDirectedSimulation} for $G$ and $G'$.  
\begin{itemize}
\item As stated in Example~\ref{example: JDJWO} for Algorithm~\ref{alg: ComputeSimulation}, before executing the statement~\ref{st: KWOAJ 5}, we have that 
	\begin{eqnarray*}
	Z & = & (\{a,b,c\} \times \{e\}) \cup \{\tuple{d,f}\} \\
	\removeZ & = & \{\tuple{a,f},\tuple{b,f},\tuple{c,f}\}.
	\end{eqnarray*}

\item Executing the ``foreach'' loop specified by the statements 5 and 6, for the iteration involved with $\tuple{r,x,y'}$, 
	\begin{itemize}
	\item if $\tuple{r,x,y'} \in \{\tuple{r,a,e}, \tuple{r,a,f}, \tuple{r,b,e}, \tuple{r,b,f}\}$, then no changes are made, 
	\item if $\tuple{r,x,y'} = \tuple{r,c,e}$, then $\tuple{c,e}$ is moved from $Z$ to $\removeZ$, 
	\item if $\tuple{r,x,y'} \in \{\tuple{r,c,f}, \tuple{r,d,e}, \tuple{r,d,f}\}$, then no changes are made.
	\end{itemize}

\item Thus, before executing the statement~\ref{st: KWOAJ 7}, we have that 
	\begin{eqnarray*}
	Z & = & \{\tuple{a,e},\tuple{b,e},\tuple{d,f}\} \\
	\removeZ & = & \{\tuple{a,f},\tuple{b,f},\tuple{c,f},\tuple{c,e}\}.
	\end{eqnarray*}

\item Executing the ``while'' loop, for the iteration involved with $\tuple{y,y'}$ extracted from $\removeZ$, 
	\begin{itemize}
	\item if $\tuple{y,y'} = \tuple{a,f}$, then no additional changes are made, 
	\item if $\tuple{y,y'} = \tuple{b,f}$, then $\tuple{a,e}$ is moved from $Z$ to $\removeZ$, 
	\item if $\tuple{y,y'} = \tuple{c,f}$, then no additional changes are made, 
	\item if $\tuple{y,y'} = \tuple{c,e}$, then $\tuple{b,e}$ is moved from $Z$ to $\removeZ$, 
	\item if $\tuple{y,y'} = \tuple{a,e}$, then no additional changes are made, 
	\item if $\tuple{y,y'} = \tuple{b,e}$, then $\tuple{d,f}$ is moved from $Z$ to $\removeZ$, 
	\item if $\tuple{y,y'} = \tuple{d,f}$, then no additional changes are made. 
	\end{itemize}

\item The algorithm returns $Z = \emptyset$. This coincides with the claim of Example~\ref{example: JDJWO 0b} that $\emptyset$ is the largest directed simulation between~$G$ and~$G'$. 
\myend
\end{itemize}
\end{example}


The following lemma is a counterpart of Lemma~\ref{lemma: JHDKA}. 
It can be proved analogously. 
 
\begin{lemma}\label{lemma: DJWYS}
The following assertions are invariants of the ``while'' loop of Algorithm~\ref{alg: ComputeDirectedSimulation}: 
\begin{enumerate}
\item The largest directed simulation between $G$ and $G'$ is a subset of $Z$.
\item For all $\tuple{r,y,x'} \in \SE \times V \times V'$, if $x \in \Prev_r(y)$ and $Z[x,x']$ holds, then \eqref{eq: CS 5b} holds.
\item For all $\tuple{r,x,y'} \in \SE \times V \times V'$, if $x' \in \PrevP_r(y')$ and $Z[x,x']$ holds, then the following counterpart of~\eqref{eq: CS 5b} holds.
	\begin{equation}\label{eq: JHFEK}
	E'(x',r,y') \leq \sup \{E(x,r,y) \mid y \in \Next_r(x) \land \tuple{y,y'} \in Z \cup \removeZ\}
	\end{equation}
\end{enumerate}
\end{lemma}

The following theorem is a counterpart of Theorem~\ref{theorem: HDKAL}. 
 
\begin{theorem}
Algorithm~\ref{alg: ComputeDirectedSimulation} always terminates and returns the largest directed simulation between the input fuzzy graphs.
\end{theorem}

This theorem can be proved analogously as done for Theorem~\ref{theorem: HDKAL}. In particular, in addition to the second assertion of Lemma~\ref{lemma: DJWYS} and Condition~\eqref{eq: CS 4}, the proof also exploits the third assertion of Lemma~\ref{lemma: DJWYS} and Condition~\eqref{eq: GFDJS}.  

Like its counterpart, Algorithm~\ref{alg: ComputeDirectedSimulation} has been formulated on an abstract level without implementation details in order to increase simplicity and facilitate understanding. We now refine this algorithm by giving implementation details so that the resulting algorithm has a complexity of order $O((m+n)n)$. In short, the new algorithm is an adapted extension of Algorithm~\ref{alg: ComputeSimulationEfficiently}. 

\begin{figure*}
	\begin{procedure}[H]
		\caption{UpdateRcPrev$'$($r,x,y'$)\label{proc: UpdateRcPrev'}\ \ // for updating $\rcPrevP$}
		
		\lIf{$\rcNext[r,x,y']$ is empty}{d := 0\label{st: HGNSA 1b}}
		\Else{
			let $y$ be the key of the last element of $\rcNext[r,x,y']$\;
			$d := E[x,r,y]$\label{st: HGNSA 4b}\;
		}
		
		\While{$\rcPrevP[r,x,y']$ is not empty\label{st: HGNSA 5b}}{
			let $x'$ be the last element of $\rcPrevP[r,x,y']$\label{st: HGNSA 6b}\; 
			\lIf{$E'[x',r,y'] \leq d$}{\Break} 
			delete from $\rcPrevP[r,x,y']$ the last element\;  
			\If{$Z[x,x']$}{
				$Z[x,x'] := \false$\;
				add $\tuple{x,x'}$ to $\removeZ$\label{st: HGNSA 11b}\;
			}
		}
	\end{procedure}
	
	\begin{procedure}[H]
		\caption{InitializeDS()\label{proc: InitializeDS}\ \ // a counterpart of $\Initialize$ for directed simulations}
		construct arrays $\Next$, $\NextP$, $\Prev$, $\PrevP$, $E$, $E'$, $\EdgeID$, $\EdgeIDp$ according to their specifications\label{st: HRLSN 1b}\; 
		\BlankLine
		
		\lForEach{$\tuple{x,x'} \in V \times V'$}{$Z[x,x'] := L(x) \leq L'(x')$\label{st: HRLSN 2b}}
		set $\removeZ$ to the empty queue\label{st: HRLSN 3b}\;
		
		\ForEach{$\tuple{r,x} \in \SE \times V$}{
			sort $\Next[r,x]$ in ascending order w.r.t.~$E[x,r,y]$ for $y \in \Next[r,x]$\;
		}  
		\ForEach{$\tuple{r,x'} \in \SE \times V'$}{
			sort $\NextP[r,x']$ in ascending order w.r.t.~$E'[x',r,y']$ for $y' \in \NextP[r,x']$\;
		}  
		\ForEach{$\tuple{r,y} \in \SE \times V$}{
			sort $\Prev[r,y]$ in ascending order w.r.t.~$E[x,r,y]$ for $x \in \Prev[r,y]$\;
		}  
		\ForEach{$\tuple{r,y'} \in \SE \times V'$}{
			sort $\PrevP[r,y']$ in ascending order w.r.t.~$E'[x',r,y']$ for $x' \in \PrevP[r,y']$\;
		}  
		\lForEach{$y \in V$ and $e' \in 0..(|E'|-1)$}{$\rcNextElemP[y,e'] := \Null$} 
		\lForEach{$y' \in V'$ and $e \in 0..(|E|-1)$}{$\rcNextElem[y',e] := \Null$} 
		
		\ForEach{$\tuple{r,y,x'} \in \SE \times V \times V'$\label{st: HRLSN 9b}}{
			construct a vector $\rcPrev[r,y,x']$ of all $x \in \Prev[r,y]$ such that $Z[x,x']$ holds, preserving the order\label{st: HRLSN 10b}\;
			construct a doubly linked list $\rcNextP[r,y,x']$ whose elements' keys are all $y' \in \NextP[r,x']$ such that $Z[y,y']$ holds, preserving the order\label{st: HRLSN 11b}\;
			\ForEach{element $u$ of $\rcNextP[r,y,x']$\label{st: HRLSN 12b}}{
				let $y'$ be the key of $u$ and let $e' = \EdgeIDp[x',r,y']$\; 
				$\rcNextElemP[y,e'] := u$\label{st: HRLSN 14b}\;
			}
		} 
		
		\ForEach{$\tuple{r,x,y'} \in \SE \times V \times V'$\label{st: HRLSN 9c}}{
			construct a vector $\rcPrevP[r,x,y']$ of all $x' \in \PrevP[r,y']$ such that $Z[x,x']$ holds, preserving the order\label{st: HRLSN 10c}\;
			construct a doubly linked list $\rcNext[r,x,y']$ whose elements' keys are all $y \in \Next[r,x]$ such that $Z[y,y']$ holds, preserving the order\label{st: HRLSN 11c}\;
			\ForEach{element $u$ of $\rcNext[r,x,y']$\label{st: HRLSN 12c}}{
				let $y$ be the key of $u$ and let $e = \EdgeID[x,r,y]$\; 
				$\rcNextElem[y',e] := u$\label{st: HRLSN 14c}\;
			}
		} 

		\lForEach{$\tuple{r,y,x'} \in \SE \times V \times V'$\label{st: HRLSN 15b}}{$\UpdateRcPrev(r,y,x')$\label{st: HRLSN 16b}}
		\lForEach{$\tuple{r,x,y'} \in \SE \times V \times V'$\label{st: HRLSN 15c}}{$\UpdateRcPrev'(r,x,y')$\label{st: HRLSN 16c}}
	\end{procedure}
\end{figure*}

Apart from the data structures described in Section~\ref{section: comp-CS2}, we also use arrays $\Next$, $\EdgeID$, $\rcNext$, $\rcNextElem$ and $\rcPrevP$, which are dual to $\NextP$, $\EdgeIDp$, $\rcNextP$, $\rcNextElemP$ and $\rcPrev$, respectively. 
For clarity, they are explicitly specified below. 

\begin{itemize}
\item For $r \in \SE$ and $x \in V$, $\Next[r,x]$ is a vector representing $\Next_r(x)$. 

\item $\EdgeID : V \times \SE \times V \to 0..(|E|-1)$ is an array used for identifying nonzero edges of $G$ with the following properties: if $\tuple{x_1,r_1,y_1}$ and $\tuple{x_2,r_2,y_2}$ are different triples from $V \times \SE \times V$ such that $E[x_1,r_1,y_1] > 0$ and $E[x_2,r_2,y_2] > 0$, then $\EdgeID[x_1,r_1,y_1] \neq \EdgeID[x_2,r_2,y_2]$. 

\item $\rcNext$ is an array such that, for $\tuple{r,x,y'} \in \SE \times V \times V'$, $\rcNext[r,x,y']$ is a doubly linked list consisting of the vertices $y \in \Next_r(x)$ such that $Z[y,y'] \lor (\tuple{y,y'} \in \removeZ)$ holds. The list is sorted in ascending order with respect to $E[x,r,y]$. The vertex contained in an element of $\rcNext[r,x,y']$ is called the key of that element. 
	
\item $\rcNextElem$ is an array with indices from $V' \times (0..(|E|-1))$ such that, if $Z[y,y'] \lor (\tuple{y,y'} \in \removeZ)$, $r \in \SE$, $x \in \Prev_r(y)$ and $e = \EdgeID[x,r,y]$, then $\rcNextElem[y',e]$ is (a reference to) the element of the doubly linked list $\rcNext[r,x,y']$ whose key is~$y$. 
	
\item $\rcPrevP$ is an array such that, for $\tuple{r,x,y'} \in \SE \times V \times V'$, $\rcPrevP[r,x,y']$ is a vector consisting of the vertices $x' \in \PrevP_r(y')$ such that the following condition holds 
	\begin{equation}
	E'[x',r,y'] \leq \sup \{E[x,r,y] \mid y \in \Next_r(x) \land (Z[y,y'] \lor (\tuple{y,y'} \in \removeZ))\}
	\end{equation}
and the vector is sorted in ascending order with respect to $E[x',r,y']$. 
\end{itemize}

\begin{algorithm}
	\caption{\ComputeDirectedSimulationEfficiently\label{alg: ComputeDirectedSimulationEfficiently}}
	\Input{finite fuzzy graphs $G = \tuple{V,E,L}$ and $G' = \tuple{V',E',L'}$.}
	\Output{the largest directed simulation between $G$ and $G'$.}
	\BlankLine
	$\InitializeDS()$\tcp*{defined on page~\pageref{proc: InitializeDS}}
	\While{$\removeZ$ is not empty\label{st: JDNAA 2b}}{
		extract $\tuple{y,y'}$ from $\removeZ$\label{st: JDNAA 3b}\;
		
		\ForEach{$r \in \SE$ and $x' \in \PrevP[r,y']$\label{st: JDNAA 4b}}{
			$e' := \EdgeIDp[x',r,y']$\label{st: JDNAA 5b}\; 
			$u := \rcNextElemP[y,e']$\; 
			delete $u$ from $\rcNextP[r,y,x']$\;
			$\rcNextElemP[y,e'] := \Null$\label{st: JDNAA 8b}\;
			$\UpdateRcPrev(r,y,x')$\label{st: JDNAA 9b}\tcp*{defined on page~\pageref{proc: UpdateRcPrev}}
		}
		
		\ForEach{$r \in \SE$ and $x \in \Prev[r,y]$\label{st: JDNAA 4c}}{
			$e := \EdgeID[x,r,y]$\label{st: JDNAA 5c}\; 
			$u := \rcNextElem[y',e]$\; 
			delete $u$ from $\rcNext[r,x,y']$\;
			$\rcNextElem[y',e] := \Null$\label{st: JDNAA 8c}\;
			$\UpdateRcPrev'(r,x,y')$\label{st: JDNAA 9c}\tcp*{defined on page~\pageref{proc: UpdateRcPrev'}}
		}
	}
	\Return the relation corresponding to $Z$\label{st: JDNAA 10b}\;
\end{algorithm}

Algorithm~\ref{alg: ComputeDirectedSimulationEfficiently} (\ComputeDirectedSimulationEfficiently) given on page~\pageref{alg: ComputeDirectedSimulationEfficiently} is a reformulation of Algorithm~\ref{alg: ComputeDirectedSimulation} using the above mentioned data structures. 
It is an adapted extension of Algorithm~\ref{alg: ComputeSimulationEfficiently} (\ComputeSimulationEfficiently). 
In particular, observe the following. 
\begin{itemize}
\item The procedure $\UpdateRcPrev'(r,x,y')$ (on page~\pageref{proc: UpdateRcPrev'}) is dual to the procedure $\UpdateRcPrev(r,y,x')$. 

\item The procedure $\InitializeDS$ (on page~\pageref{proc: InitializeDS}) is a counterpart of the procedure $\Initialize$. 
It initializes not only $\EdgeIDp$, $\rcNextP$, $\rcNextElemP$ and $\rcPrev$ but also $\EdgeID$, $\rcNext$, $\rcNextElem$ and $\rcPrevP$ dually. The postfix ``DS'' relates to ``directed simulation''. 

\item The statements~\ref{st: JDNAA 2b}-\ref{st: JDNAA 9b} of Algorithm~\ref{alg: ComputeDirectedSimulationEfficiently} are the same as the statements~\ref{st: JDNAA 2}-\ref{st: JDNAA 9} of Algorithm~\ref{alg: ComputeSimulationEfficiently}. 

\item The statements~\ref{st: JDNAA 4c}-\ref{st: JDNAA 9c} of Algorithm~\ref{alg: ComputeDirectedSimulationEfficiently} are dual to the statements~\ref{st: JDNAA 4b}-\ref{st: JDNAA 9b}.
\end{itemize}

We have implemented Algorithm~\ref{alg: ComputeDirectedSimulationEfficiently} in C++ and shared the code~\cite{compCSfFS-impl} to allow experiments with it.  

The following lemma is a counterpart of Lemma~\ref{lemma: JHDKA 2}. 
It can be proved analogously. 

\begin{lemma}\label{lemma: DJWYS 2}
The following assertions are invariants of the ``while'' loop of Algorithm~\ref{alg: ComputeDirectedSimulationEfficiently}: 
\begin{enumerate}
\item The data structures $\rcNext$, $\rcNextP$, $\rcNextElem$, $\rcNextElemP$, $\rcPrev$ and $\rcPrevP$ satisfy their specifications. 

\item The largest direct simulation between $G$ and $G'$ is a subset of $\{\tuple{x,x'} \in V \times V' \mid Z[x,x']\}$.

\item For all $\tuple{r,y,x'} \in \SE \times V \times V'$, if $x \in \Prev_r(y)$ and $Z[x,x']$ holds, then $x \in \rcPrev[r,y,x']$. 

\item For all $\tuple{r,x,y'} \in \SE \times V \times V'$, if $x' \in \PrevP_r(y')$ and $Z[x,x']$ holds, then $x' \in \rcPrevP[r,x,y']$. 
\end{enumerate}
\end{lemma}

The following theorem is a counterpart of Theorem~\ref{theorem: JHLWA}. It can be proved analogously, using Lemma~\ref{lemma: DJWYS 2}
 instead of Lemma~\ref{lemma: JHDKA 2}. 
 
\begin{theorem}\label{theorem: JHLWA 2}
Algorithm~\ref{alg: ComputeDirectedSimulationEfficiently} runs in time $O((m+n)n)$ and returns the largest directed simulation between given finite fuzzy graphs $G = \tuple{V,E,L}$ and $G' = \tuple{V',E',L'}$, where $n = |V| + |V'|$ and $m = |E| + |E'|$. 
\end{theorem}

\section{Adapting to Fuzzy Automata}
\label{sec: HFKDS}

\red{In this section, we adapt Algorithms~\ref{alg: ComputeSimulationEfficiently} and~\ref{alg: ComputeDirectedSimulationEfficiently} to computing the largest crisp simulation and the largest crisp directed simulation between two finite fuzzy automata.} 

A {\em fuzzy automaton} over a (finite) alphabet $\Sigma$ is a tuple $A = \tuple{Q,\delta,\sigma,\tau}$, where $Q$ is a non-empty set of {\em states}, $\delta: Q \times \Sigma \times Q \to [0,1]$ is called the {\em fuzzy transition function}, $\sigma: Q \to [0,1]$ the {\em fuzzy set of initial states}, and $\tau: Q \to [0,1]$ the {\em fuzzy set of terminal states} (cf.~\cite{CiricIDB12}). It is {\em finite} if $Q$ is finite.

Let $A = \tuple{Q,\delta,\sigma,\tau}$ and $A' = \tuple{Q',\delta',\sigma',\tau'}$ be fuzzy automata over an alphabet $\Sigma$. 
A relation $Z \subseteq Q \times Q'$ is a {\em (crisp) simulation} between $A$ and $A'$ if the following conditions hold (for all possible values of the free variables), where $\to$ and $\land$ denote the usual crisp logical connectives: 
\begin{eqnarray}
\sigma(x) > 0 & \to & \E x' \in Q'\, (Z(x,x') \land \sigma(x) < \sigma'(x')) \label{eq: YRJSK-1} \\
Z(x,x') \land \delta(x,r,y) > 0 & \to & \E y' \in Q'\, (Z(y,y') \land \delta(x,r,y) < \delta'(x',r,y')) \label{eq: YRJSK-2} \\
Z(x,x') \land \tau(x) > 0 & \to & \tau(x) < \tau'(x'). \label{eq: YRJSK-3} 
\end{eqnarray}

The above definition of crisp simulations between fuzzy automata is an adaptation of the definition of ``forward (fuzzy) simulation'' between fuzzy automata introduced by {\'C}iri{\'c} {\em et at.} in~\cite{CiricIDB12}.  

A relation $Z \subseteq Q \times Q'$ is a {\em (crisp) directed simulation} between $A$ and $A'$ if Conditions~\eqref{eq: YRJSK-1}-\eqref{eq: YRJSK-3} and the following one hold (for all possible values of the free variables):
\begin{eqnarray}
Z(x,x') \land \delta'(x',r,y') > 0 & \to & \E y \in Q\, (Z(y,y') \land \delta'(x',r,y') < \delta(x,r,y)) \label{eq: YRJSK-4}
\end{eqnarray}

Given a fuzzy automaton $A = \tuple{Q,\delta,\sigma,\tau}$ over an alphabet $\Sigma$, we define the {\em fuzzy graph corresponding to $A$} to be the fuzzy graph $G = \tuple{V,E,L}$ using $\SE = \Sigma$ and $\SV = \{i,f\}$ such that:
\begin{itemize}
	\item $V = Q \cup \{v_i,v_f\}$, where $v_i$ and $v_f$ are new vertices, with $v_i$ standing for the new unique initial state, and $v_f$ the new unique final state, 
	\item $L(v_i)(i) = L(v_f)(f) = 1$, $L(v_i)(f) = L(v_f)(i) = 0$, and $L(x)(i) = L(x)(f) = 0$ for $x \in Q$ (thus, $i \in \SV$ is used to identify $v_i$ and $f \in \SV$ is used to identify $v_f$), 
	\item for every $r \in \SE$, $x,y \in Q$ and $z \in V$: 
	\begin{itemize} 
		\item $E(x,r,y) = \delta(x,r,y)$, 
		\item $E(v_i,r,x) = \sigma(x)$ and $E(x,r,v_f) = \tau(x)$, 
		\item $E(z,r,v_i) = E(v_f,r,z) = E(v_i,r,v_f) = 0$. 
	\end{itemize}
\end{itemize}

The above definition is a counterpart of the definition of the fuzzy interpretation (in description logic) corresponding to a~fuzzy automaton~\cite{TFS2020}. 
The following lemma can be proved in a straightforward way. 

\begin{lemma}\label{lemma: HJHSA}
Let $A = \tuple{Q,\delta,\sigma,\tau}$ and $A' = \tuple{Q',\delta',\sigma',\tau'}$ be fuzzy automata over the same alphabet. 
Let $G = \tuple{V,E,L}$ and $G' = \tuple{V',E',L'}$ be the fuzzy graphs corresponding to $A$ and $A'$, respectively. 
Let $v_i, v_f \in V$ and $v'_i, v'_f \in V'$ be the vertices such that $L(v_i)(i) = L(v_f)(f) = L'(v'_i)(i) = L'(v'_f)(f) = 1$. If $Z \subseteq Q \times Q'$, then $Z$ is a simulation (respectively, directed simulation) between $A$ and $A'$ iff $Z \cup \{\tuple{v_i,v'_i},\tuple{v_f,v'_f}\}$ is a simulation (respectively, directed simulation) between $G$ and $G'$.  
\end{lemma}

\newcommand{\AlgorithmSimFA}{Algorithm~\ref{alg: ComputeSimulationEfficiently}'\xspace}
\newcommand{\AlgorithmDirSimFA}{Algorithm~\ref{alg: ComputeDirectedSimulationEfficiently}'\xspace}

{\markRed 
Let \AlgorithmSimFA (respectively, \AlgorithmDirSimFA) be the algorithm that, given finite fuzzy automata $A = \tuple{Q,\delta,\sigma,\tau}$ and $A' = \tuple{Q',\delta',\sigma',\tau'}$ over the same alphabet~$\Sigma$, computes the largest simulation (respectively, directed simulation) between $A$ and $A'$ as follows:
\begin{enumerate}
\item construct the fuzzy graph $G = \tuple{V,E,L}$ that corresponds to $A$ and let $v_i, v_f \in V$ be the added vertices (with $L(v_i)(i) = L(v_f)(f) = 1$); 
\item construct the fuzzy graph $G' = \tuple{V',E',L'}$ that corresponds to $A'$ and let $v'_i, v'_f \in V'$ be the added vertices (with $L'(v'_i)(i) = L'(v'_f)(f) = 1$); 
\item run Algorithm~\ref{alg: ComputeSimulationEfficiently} (respectively, Algorithm~\ref{alg: ComputeDirectedSimulationEfficiently}) to compute the largest simulation (respectively, directed simulation) $Z$ between $G$ and $G'$;
\item return $Z - \{\tuple{v_i,v'_i},\tuple{v_f,v'_f}\}$.
\end{enumerate}

\begin{theorem}
\AlgorithmSimFA (respectively, \AlgorithmDirSimFA) returns the largest simulation (respectively, directed simulation) between the given finite fuzzy automata $A = \tuple{Q,\delta,\sigma,\tau}$ and $A' = \tuple{Q',\delta',\sigma',\tau'}$, which are over the same alphabet $\Sigma$. Its complexity is of order $O((m+n)n)$, where $n = |Q| + |Q'|$ and $m = |\delta| + |\delta'|$, with 
\begin{eqnarray*}
|\delta| & = & |\{\tuple{x,r,y} \in Q \times \Sigma \times Q : \delta(x,r,y) > 0\}| \\
|\delta'| & = & |\{\tuple{x',r,y'} \in Q' \times \Sigma \times Q' : \delta(x',r,y') > 0\}|.
\end{eqnarray*} 
\end{theorem}

This theorem follows immediately from Lemma~\ref{lemma: HJHSA} and Theorem~\ref{theorem: JHLWA} (respectively, Theorem~\ref{theorem: JHLWA 2}) and the fact that the first two steps of the algorithm (for constructing $G$ and $G'$) can be done in time $O(m+n)$.
}

\section{Conclusions}
\label{sec: conc}

As far as we know, before the current work there were no algorithms directly formulated for computing the largest crisp simulation and the largest crisp directed simulation between two finite fuzzy graph-based structures. One can try to adapt the algorithm of computing the greatest right/left invariant Boolean quasi-order matrix of a~weighted automaton over an additively idempotent semiring, given by Stanimirovi{\'c} {\em et at.} in~\cite{StanimirovicSC2019}, to obtain algorithms for those tasks. However, the complexity order $O(n^5)$ is too high. \red{Crisping a given finite fuzzy graph and then applying one of the algorithms given in~\cite{BloomP95,HenzingerHK95,CompBSDLP} to the obtained crisp graph also results in an algorithm with a high complexity order, $O(l(m+n)n)$, where $l$ is the number of fuzzy values occurring in the specification of the input graph and, in the worst case, can be $m+n$.} 

In this article, we have given efficient algorithms with the complexity $O((m+n)n)$ for computing the largest crisp simulation and the largest crisp directed simulation between two finite fuzzy labeled graphs, where $n$ is the number of vertices and $m$ is the number of nonzero edges of the input fuzzy graphs. \red{We have also adapted them to computing the largest crisp simulation and the largest crisp directed simulation between two finite fuzzy automata.} 


\bibliography{BSfDL}

\begin{thebibliography}{10}

\bibitem{BRV2001}
P.~Blackburn, M.~de~Rijke, and Y.~Venema.
\newblock {\em Modal Logic}.
\newblock Number~53 in Cambridge Tracts in Theoretical Computer Science.
  Cambridge University Press, 2001.

\bibitem{BloomP95}
B.~Bloom and R.~Paige.
\newblock Transformational design and implementation of a new efficient
  solution to the ready simulation problem.
\newblock {\em Sci. Comput. Program.}, 24(3):189--220, 1995.

\bibitem{CaoCK11}
Y.~Cao, G.~Chen, and E.E. Kerre.
\newblock Bisimulations for fuzzy-transition systems.
\newblock {\em {IEEE} Trans. Fuzzy Systems}, 19(3):540--552, 2011.

\bibitem{CaoSWC13}
Y.~Cao, S.X. Sun, H.~Wang, and G.~Chen.
\newblock A behavioral distance for fuzzy-transition systems.
\newblock {\em {IEEE} Trans. Fuzzy Systems}, 21(4):735--747, 2013.

\bibitem{MicicJS18}
I.~{\'C}iri{\'c}, Z.~Jan\v{c}i{\'c}, and S.~Stanimirovi{\'c}.
\newblock Computation of the greatest right and left invariant fuzzy
  quasi-orders and fuzzy equivalences.
\newblock {\em Fuzzy Sets and Systems}, 339:99--118, 2018.

\bibitem{CiricIDB12}
M.~{\'C}iri{\'c}, J.~Ignjatovi{\'c}, N.~Damljanovi{\'c}, and M.~Ba\u{s}ic.
\newblock Bisimulations for fuzzy automata.
\newblock {\em Fuzzy Sets and Systems}, 186(1):100--139, 2012.

\bibitem{CiricIJD12}
M.~{\'C}iri{\'c}, J.~Ignjatovi{\'c}, I.~Jan\u{c}i{\'c}, and N.~Damljanovi{\'c}.
\newblock Computation of the greatest simulations and bisimulations between
  fuzzy automata.
\newblock {\em Fuzzy Sets and Systems}, 208:22--42, 2012.

\bibitem{BSDL-P-LOGCOM}
A.R. Divroodi and L.A. Nguyen.
\newblock On directed simulations in description logics.
\newblock {\em J. Log. Comput.}, 27(7):1955--1986, 2017.

\bibitem{Fan15}
T.-F. Fan.
\newblock Fuzzy bisimulation for {G\"{o}del} modal logic.
\newblock {\em {IEEE} Trans. Fuzzy Systems}, 23(6):2387--2396, 2015.

\bibitem{ai/FanL14}
T.{-}F. Fan and C.{-}J. Liau.
\newblock Logical characterizations of regular equivalence in weighted social
  networks.
\newblock {\em Artif. Intell.}, 214:66--88, 2014.

\bibitem{HennessyM85}
M.~Hennessy and R.~Milner.
\newblock Algebraic laws for nondeterminism and concurrency.
\newblock {\em Journal of the ACM}, 32(1):137--161, 1985.

\bibitem{HenzingerHK95}
M.R. Henzinger, T.A. Henzinger, and P.W. Kopke.
\newblock Computing simulations on finite and infinite graphs.
\newblock In {\em Proceedings of FOCS'1995}, pages 453--462. {IEEE} Computer
  Society, 1995.

\bibitem{IgnjatovicCS15}
J.~Ignjatovi{\'c}, M.~{\'C}iri{\'c}, and I.~Stankovi{\'c}.
\newblock Bisimulations in fuzzy social network analysis.
\newblock In {\em Proceedings of IFSA-EUSFLAT-15}. Atlantis Press, 2015.

\bibitem{KurtoninaR97}
N.~Kurtonina and M.~de~Rijke.
\newblock Simulating without negation.
\newblock {\em J. Log. Comput.}, 7(4):501--522, 1997.

\bibitem{CompBSDLP}
L.A. Nguyen.
\newblock Computing bisimulation-based comparisons.
\newblock {\em Fundam. Inform.}, 157(4):385--401, 2018.

\bibitem{Nguyen-TFS2019}
L.A. Nguyen.
\newblock Bisimilarity in fuzzy description logics under the {Zadeh} semantics.
\newblock {\em {IEEE} Trans. Fuzzy Systems}, 27(6):1151--1161, 2019.

\bibitem{jBSfDL3}
L.A. Nguyen.
\newblock Characterizing crisp notions of bisimulation, simulation and directed
  simulation in fuzzy description logics.
\newblock Manuscript in a final stage of preparation, to be submitted in
  September 2020, 2020.

\bibitem{compCSfFS-impl}
L.A. Nguyen.
\newblock An implementation in {C++} of the agorithms provided in the current
  paper.
\newblock Available at
  \url{https://drive.google.com/drive/folders/1vfKtkVqt1PDc73fXXsrVZSwCBCDGYrbu?usp=sharing}
  and \url{https://www.mimuw.edu.pl/~nguyen/compCSfFS/CompBSDLP-prog.zip},
  2020.

\bibitem{jBSfDL2}
L.A. Nguyen, Q.-T. Ha, N.T. Nguyen, T.H.K. Nguyen, and T.-L. Tran.
\newblock Bisimulation and bisimilarity for fuzzy description logics under the
  {G\"odel} semantics.
\newblock {\em Fuzzy Sets and Systems}, 388:146--178, 2020.

\bibitem{minimization-by-fBS}
L.A. Nguyen and N.-T. Nguyen.
\newblock Minimizing interpretations in fuzzy description logics under the
  {G\"{o}del} semantics by using fuzzy bisimulations.
\newblock {\em Journal of Intelligent and Fuzzy Systems}, 37(6):7669--7678,
  2019.

\bibitem{TFS2020}
L.A. {Nguyen} and D.X. {Tran}.
\newblock Computing fuzzy bisimulations for fuzzy structures under the
  {G\"odel} semantics.
\newblock {\em IEEE Transactions on Fuzzy Systems, {\em DOI:
  10.1109/TFUZZ.2020.2985000}, {\em
  {https://ieeexplore.ieee.org/document/9056459/}}}, 2020.

\bibitem{PaigeT87}
R.~Paige and R.E. Tarjan.
\newblock Three partition refinement algorithms.
\newblock {\em SIAM J. Comput.}, 16(6):973--989, 1987.

\bibitem{Park81}
D.M.R. Park.
\newblock Concurrency and automata on infinite sequences.
\newblock In Peter Deussen, editor, {\em Proceedings of the 5th GI-Conference},
  volume 104 of {\em LNCS}, pages 167--183. Springer, 1981.

\bibitem{Sangiorgi09}
D.~Sangiorgi.
\newblock On the origins of bisimulation and coinduction.
\newblock {\em {ACM} Trans. Program. Lang. Syst.}, 31(4):15:1--15:41, 2009.

\bibitem{StanimirovicSC2019}
S.~Stanimirovi{\'c}, A.~Stamenkovi{\'c}, and M.~{\'C}iri{\'c}.
\newblock Improved algorithms for computing the greatest right and left
  invariant boolean matrices and their application.
\newblock {\em Filomat}, 33(9):2809--2831, 2019.

\bibitem{vBenthem76}
J.~van Benthem.
\newblock {\em Modal Correspondence Theory}.
\newblock PhD thesis, Mathematisch Instituut \& Instituut voor
  Grondslagenonderzoek, University of Amsterdam, 1976.

\bibitem{vBenthem83}
J.~van Benthem.
\newblock {\em Modal Logic and Classical Logic}.
\newblock Bibliopolis, Naples, 1983.

\bibitem{vBenthem84}
J.~van Benthem.
\newblock Correspondence theory.
\newblock In D.~Gabbay and F.~Guenther, editors, {\em Handbook of Philosophical
  Logic, Volume {II}}, pages 167--247. Reidel, Dordrecht, 1984.

\bibitem{DBLP:journals/fss/WuCBD18}
H.~Wu, Y.~Chen, T.{-}M. Bu, and Y.~Deng.
\newblock Algorithmic and logical characterizations of bisimulations for
  non-deterministic fuzzy transition systems.
\newblock {\em Fuzzy Sets Syst.}, 333:106--123, 2018.

\bibitem{DBLP:journals/fss/WuD16}
H.~Wu and Y.~Deng.
\newblock Logical characterizations of simulation and bisimulation for fuzzy
  transition systems.
\newblock {\em Fuzzy Sets Syst.}, 301:19--36, 2016.

\end{thebibliography}
\bibliographystyle{plain}

\end{document}